%% file: bdcat-main.tex
\documentclass[10pt, conference, letterpaper]{IEEEtran}
\makeatletter
\def\ps@headings{%
\def\@oddhead{\mbox{}\scriptsize\rightmark \hfil \thepage}%
\def\@evenhead{\scriptsize\thepage \hfil \leftmark\mbox{}}%
\def\@oddfoot{}%
\def\@evenfoot{}}
\makeatother
\usepackage[utf8]{inputenc}
\usepackage[T1]{fontenc}
\usepackage{url}
\usepackage{algorithm}
\usepackage{algorithmic}
\usepackage{graphicx}
\usepackage{amsfonts}
\usepackage{amsmath}
\usepackage{amsthm}
\usepackage{subfig}
\usepackage{enumerate}
\usepackage{multirow}
\usepackage{eqparbox}
\usepackage{epstopdf}
\usepackage[noadjust]{cite}
\usepackage{xcolor}
\usepackage{bbm}
\usepackage{nicefrac}

\usepackage[normalem]{ulem}

\newcommand{\calO}{\mathcal{O}}
\newcommand{\E}{\mathbb{E}}
\newcommand{\X}{\mathbf{X}}
\newcommand{\dist}{d_X}
\newcommand{\distPre}{d_{PRE}}
\newcommand{\id}{\mathit{id}}
\newcommand{\embed}{\mathbf{A}}
\newcommand{\stab}{\mathbf{S}}

\newcommand{\add}{\mathit{ca}}
\newcommand{\m}{M}
\newcommand{\closest}{\mathbf{B}}
\newcommand{\integers}{Ic}

\newcommand{\clemens}[1]{{\color{green}#1}}
\newcommand{\stef}[1]{#1}
\newcommand{\stefout}[2]{#2}

\newcommand{\longonly}[1]{#1}
\newcommand{\shortonly}[1]{}

\newtheorem{theorem}{Theorem}[section]
\newtheorem{lemma}[theorem]{Lemma}

\newtheorem{corollary}[theorem]{Corollary}
\newtheorem{prop}[theorem]{Proposition}
\newtheorem{definition}[theorem]{Definition}

\author{
  \IEEEauthorblockN{Stefanie Roos\IEEEauthorrefmark{3}, Martin Byrenheid\IEEEauthorrefmark{2}, Clemens Deusser\IEEEauthorrefmark{2}, Thorsten Strufe\IEEEauthorrefmark{2}}
  \IEEEauthorblockA{\IEEEauthorrefmark{3}University of Waterloo\\
  sroos@uwaterloo.ca
  }
  \IEEEauthorblockA{\IEEEauthorrefmark{2}TU Dresden\\
  \{\texttt{firstname.lastname}\}@tu-dresden.de
  }
}

\title{\shortonly{BD-CAT: }Balanced Dynamic Content Addressing in Trees\longonly{\footnotemark{*}}}

\begin{document}
\maketitle \longonly{\footnotetext{Extended Version of 'BD-CAT: Balanced Dynamic Content Addressing in Trees’, INFOCOM 2017}}

\begin{abstract}

Balancing the load in content addressing schemes for route-restricted networks represents a challenge with a wide range of applications.
Solutions based on greedy embeddings maintain minimal state information and enable efficient routing, but any such solutions currently result in either imbalanced content addressing, overloading individual nodes, or are unable to efficiently account for network dynamics.


In this work, we propose a greedy embedding in combination with a content addressing scheme that provides balanced content addressing while at the same time enabling efficient stabilization in the presence of network dynamics.
We point out the trade-off between stabilization complexity and maximal permitted imbalance when deriving upper bounds on both metrics for two variants of the proposed algorithms.
Furthermore, we substantiate these bounds through a simulation study based on both real-world and synthetic data.   


\end{abstract}

\input{intro}

\input{related}

\input{notation}

\input{algorithm}
\longonly{\input{analysis}}
\input{eval}
\input{conclusion}

\section*{Acknowledgements}
This work in parts was supported by DFG through the CRC HAEC and the
Cluster of Excellence cfaed. 

\bibliographystyle{unsrt}
\bibliography{bdcat-main}

\end{document}

%% file: intro.tex
\section{Introduction}
Efficiently routing packets while maintaining little to no state information is a fundamental problem of networking.
The issue concerns Internet routing, in particular content-centric networking \cite{roos2014enhancing}, as well as dynamic networks such as wireless sensor networks \cite{jiang2014distributed} and Friend-to-Friend (F2F) overlays in the manner of Freenet \cite{clarke2010private}.
The routing configuration is frequently adapted to implement content addressing, where the node identifier (or: address) is used to determine the allocation of resources to specific nodes.
This scenario typically makes the configuration and routing particularly difficult, as the nodes are expected to exhibit extensive dynamics in terms of joining and leaving the system.

Greedy embeddings guarantee the success of stateless greedy routing and thus facilitate efficient communication \cite{PapadimitriouRatajczak05}. 
All existing distributed greedy embeddings are based on creating a spanning tree and subsequently assigning identifiers to each node. 
Some embedding algorithms can account for topology changes without a complete recomputation of the local state \cite{CvetkovskiCrovella09,HerzenEtAl11}.
In contrast to structured P2P overlays, greedy embeddings do not require the ability to change the network topology, making them suitable for all of the above scenarios. 

Implementing content addressing on greedy embeddings, however, faces several challenges. 
The current proposals are either unable to assign content in a fair manner \cite{Kleinberg07,hofer2013greedy}, are unable to deal with dynamics \cite{roos2014enhancing}, or considerably reduce the efficiency by establishing an additional overlay \cite{roos2016anonymous}.

We aim to realize fair resource allocation in terms of a balanced content addressing in such trees in dynamic environments. 
In other words, we require an embedding algorithm in combination with a content addressing scheme such that i) the overhead of stabilization after node arrivals or departures is low on average and ii) the content addressing is balanced, i.e., the fraction of content assigned to a node should not considerably exceed its share of the overall storage capacity.  

In this paper, we propose to assign each content an address in the form of a vector of keyed hashes. 
Similarly, we assign node addresses in the form of vectors.
The vector encodes the part of the namespace (in our case: hashes of content) that is allocated to the respective node, and each component of the vector contains a tuple indicating ranges within the namespace.
Node addresses are only changed if topology adaptations result in nodes being responsible for more addresses than the current upper bound permits. 

Our algorithm assigns at most $\calO\left(\frac{\log n}{n} \right)$ of the content to a node at any time if the tree depth is $\calO(\log n)$. Thus, the asymptotic bound matches the bound for DHTs \cite{malkhi2002viceroy}. 
Furthermore, the expected communication complexity for stabilization after a node join or departure is $\calO(polylog(n))$ if the expected number of siblings, i.e., the nodes with the same parent, is bound polylog in $n$. 
Otherwise, if such a bound on the number of siblings does not exist, the use of virtual binary trees allows us to achieve polylog complexity nevertheless, at the price of storing up to $\calO\left(\frac{\log^2 n}{n}\right)$ of the content on one node. 
We perform a simulation study based on real-world churn traces and topologies of several thousands of nodes to quantify the stabilization overhead and the balance of the content addressing in exemplary scenarios. 
Our results indicate that i) the average stabilization overhead is reduced to less than 3\% of the overhead of a complete re-embedding, and ii) the content addressing exhibits a similar or even better fairness than common content addressing schemes such as DHTs.  


%% file: related.tex
\section{Related Work}
\label{sec:related}
Greedy embeddings assign coordinates to nodes in a graph such that nodes can route messages based only on the coordinates of their neighbors.
Generally, an embedding algorithm computes such an embedding by first constructing a spanning tree and then assigning coordinates starting from the root. 
Parents assign their children coordinates based on their own coordinate.
In this manner, greedy embeddings realize efficient routing in any connected graph while maintaining very little state information. 
 
During the last years, a multitude of embedding algorithms has been developed, using coordinates from hyperbolic \cite{Kleinberg07,CvetkovskiCrovella09,EppsteinGoodrich11-SuccinctHyperEmbed,Maymounkov06}, Euclidean \cite{Maymounkov06,WestphalPei09}, or custom-metric
spaces \cite{HerzenEtAl11,ZhangEtAL11}.   
However, the problem of content addressing is mostly disregarded, with a few notable exceptions discussed in the following.

For instance, the authors of \cite{Kleinberg07} and \cite{NewsomeSong03} show that their embedding allows for content addressing. 
However, neither consider the fraction of addresses, and thus content, assigned to individual nodes. 
When applying \cite{Kleinberg07} on autonomous system (AS) topologies of several hundreds of nodes, the algorithm allocates more than 90\% of all content to one node \cite{roos2014enhancing}.  

To the best of our knowledge, \cite{hofer2013greedy} first considers load balancing in terms of content addressing for greedy embeddings. 
The authors design Prefix Embedding, an embedding algorithm for Friend-to-Friend (F2F) overlays, and evaluate how their design performs when applied for content storage and retrieval in such route-restricted overlays.   
Their simulation indicate a high imbalance with regard to the fraction of stored content, sometimes assigning more than 50\% of all content to a single node in an overlay of tens of thousands of nodes. 
Roos et al. \cite{roos2014enhancing} inversely adapt the addressing scheme for the content and assign \emph{topology-aware keys} to files, i.e., the address of a file depends on the structure of the spanning tree. In this manner, the expected fraction of files with an address in a certain range corresponds to the fraction of node coordinates in this range. 
Though the content addressing is indeed balanced, the approach requires that the spanning tree is globally known. Furthermore, network dynamics result in constant changes of node coordinates and file addresses, which make indexing of addresses difficult.

In contrast, \cite{roos2016anonymous} circumvents the problem of content addressing directly on the embedding by establishing an additional structured overlay on top of it. However, in this manner, they decrease the efficiency of the routing by a factor of about 4. 

In summary, balanced content addressing in embeddings for dynamic networks without global topology information is an open problem.
In the following, we propose and evaluate a solution.

%% file: notation.tex
\section{Problem Formalization}
\label{sec:formal}

In this section, we introduce basic notation and formally express our goals.
The key terms we need to define are those of a (greedy) embedding, a content addressable storage, and a stabilization algorithm for such a structure. 

\subsection{Graphs and Embeddings}

Throughout the paper, we rely on a graph $G=(V,E)$ with nodes $V$ and edges $E \subset V\times V$.
For simplicity, we restrict our analysis to graphs that are bidirectional, i.e., $(u,v) \in E$ iff $(v,u)\in E$ for all $u,v \in V$, and connected, i.e., there exist $w_0=u,w_1, \ldots, w_{l-1}, w_l=v$ such that $(w_{i-1},w_{i})\in E$ for all $i=1\ldots l$.
Furthermore, we denote the set of \emph{neighbors} of $v \in V$ by $N(v)=\{u \in V: (u,v) \in E\}$.

A \emph{spanning tree} is defined as a subgraph $T_G=(V, E^T)$ of $G$ such that $T_G$ is connected and $E^T \subset E$ is of minimal size. 
In a spanning tree, there exists exactly one path between every source node $s$ and destination $e$. 
A \emph{rooted spanning tree} is a spanning tree $T_G$ with a distinguished element $r \in V$, the root. We express the positions of nodes in the spanning tree with regard to the root. The \emph{level} or \emph{depth} of a node $u$ is the length of the unique path between $u$ and the root in the spanning tree. Furthermore, the depth of the tree is the maximal depth over all nodes. 
In addition, the relation of two nodes $u,v \in V$ can be expressed in relation to the root. If $u$ is included in the unique path between the root $r$ and $v$, $u$ is an \emph{ancestor} of $v$ and $v$ a \emph{descendant} of $u$.
Furthermore, if the edge $(v,u) \in E^T$, $u$ is the parent of $v$ and $v$ a child of $u$.
Children of the same node are called \emph{siblings}.
\longonly{Embeddings usually rely on rooted spanning trees to assign coordinates to nodes.

\begin{definition}
\label{def:embeddings}
A \emph{(graph) embedding} on a graph $G=(V,E)$ is a function $\id: V \rightarrow \X$ into a metric space $(\X, \dist)$.
We call $\id(u)$ the \emph{coordinate} or \emph{address} of $u$.
Consider a pair of distinct nodes $u,v \in V$\clemens{, $(u,v) \notin E$}. 
The embedding $\id$ is called \emph{greedy} if for all such pairs, $u$ has a neighbor $w \in N(u)$ with $\dist(\id(w), \id(v))< \dist(\id(u), \id(v))$.  
The algorithm $\mathbf{A}$ for deriving the embedding $\id$ is called an \emph{embedding algorithm}.
\end{definition}
For brevity, we generally write distance of $u$ and $v$ to refer to the distance of their coordinates. The above definition holds for any distance $\dist: \X \times \X \rightarrow \mathbb{R}$. We introduce realizations for $\X$ and $\dist$ in Section \ref{sec:algo}. 
Then, an equivalent definition of a greedy embedding is the guaranteed successful termination of the standard greedy routing algorithm, which specifies that each node along the path between source and destination forwards the message to the closest neighbor to the destination.
If the coordinate assignment $\id$ relies on the previous construction of a rooted spanning tree, we call $\id$ a \emph{tree-based embedding} or \emph{tree-based greedy embedding} if $\id$ is greedy.
So, greedy embeddings allow the discovery of a node by a standard greedy algorithm. However, there is little work on how to store and retrieve content based on such an embedding.
}

\shortonly{A \emph{(graph) embedding} on a graph $G=(V,E)$ is a function $\id: V \rightarrow \X$ into a metric space $(\X, \dist)$.
We call $\id(u)$ the \emph{coordinate} or \emph{address} of $u$.
Consider a pair of distinct nodes $u,v \in V$. 
The embedding $\id$ is called \emph{greedy} if for all such pairs, $u$ has a neighbor $w \in N(u)$ with $\dist(\id(w), \id(v))< \dist(\id(u), \id(v))$.  
\stef{The algorithm $\mathbf{A}$ for deriving the embedding $\id$ is called an \emph{embedding algorithm}.} 
For brevity, we generally write distance of $u$ and $v$ to refer to the distance of their coordinates.
Then, an equivalent definition of a greedy embedding is the guaranteed successful termination of the standard greedy routing algorithm, which specifies that each node along the path between source and destination forwards the message to the closest neighbor to the destination.
If the coordinate assignment $\id$ relies on the previous construction of a rooted spanning tree, we call $\id$ a \emph{tree-based embedding} or \emph{tree-based greedy embedding} if $\id$ is greedy.}

\subsection{Balanced Content Addressing} 

Content addressing generally refers to a deterministic addressing scheme for content. In the context of distributed systems, content addressing implies mapping content to nodes based on node coordinates and content addresses.
Here, we map content to the node closest to the content's address.
The scenario can be easily generalized such that content is stored on $k>1$ nodes by e.g., storing content on the closest $k$ nodes or using $k$ different addresses for each file \cite{chen2008insight}. 

In order to allow for content to be stored on closest nodes, we first need to extend the notion of a greedy embedding.

\begin{definition}
\label{def:generalGreedy}
Let $\id: V \rightarrow \X$ be a greedy embedding on a graph $G$ and $\X' \subset \X$ a countable address space, and $\add: C \rightarrow \X'$ an \emph{addressing function} for a set of content $C$.
Then $\id$ is called a \emph{content addressable greedy} embedding if i) $|\m(x')|= |argmin_{v \in V}\{\dist(id(v),x')\}|=1$ for all $x' \in X'$, and 
ii) $\forall x' \in \X', \forall v \in V, v \notin \m(x), \exists w \in N(v): \dist(\id(w),x) < \dist(\id(v),x)$.
For a graph $G=(V,E)$ with such an embedding, the tuple $(G,\id,C,\add)$ is called a \emph{content addressable storage}.
\end{definition} 

Definition \ref{def:generalGreedy} guarantees that greedy routing terminates at the closest node to an address $x$. Thus, nodes can store and retrieve files using greedy routing. 
However, Definition \ref{def:generalGreedy} does not demand that the content is distributed on the nodes in a balanced manner. Thus, we now characterize the notion of balanced or fair content addressing.
\begin{definition}
\label{def:fbalance}
Let $(G,\id,C, \add)$ be a content addressable storage.
Furthermore, $\forall v \in V$ let $\closest(v)=\{x \in \X': \forall w \in V \dist(\id(v),x)\leq \dist(\id(w),x)\}$ be the set of coordinates in $\X'$ closest to $v$, and $\mu$ be the normalized point measure, i.e., $\mu$ maps a subset $E$ of $\X'$ to the fraction of coordinates contained in $E$.  
$(G,\id,C,\add)$ is said to be \emph{$(f,\delta)$-balanced} for a real-valued factor $f \geq 1$ if 
\begin{align}
\label{eq:fbalance}
\forall v \in V, \mu(\closest(v)) \leq f\cdot \frac{1}{|V|}+\delta.
\end{align}
An embedding algorithm $\mathbf{A}$ is called \emph{$(f,\delta)$-balanced} if it generates embeddings $\id$ such that the  content addressable storage $(G,\id,C,\add)$ is $(f,\delta)$-balanced.
\end{definition}
Essentially, Definition \ref{def:fbalance} states that the expected fraction of content assigned to a node should at most be $f$ times the average content assigned to each node.
A well-known example for balanced content addressing on freely adaptable topologies are DHTs.
In DHTs, file addresses correspond to $b$-bit hashes of either the file's name, description, or content.
DHTs are  $(\calO(\log n), 0)$-balanced \cite{malkhi2002viceroy}.

We now shortly motivate some details in Definition \ref{def:fbalance}. 
The additional term $\delta$ is assumed to be small in comparison to $1/n$. Its purpose is mainly to compensate for rounding errors emerging from the fact that, for a finite $X$, $n$ most likely does not evenly divide $|X|$.
The use of the normalized point measure $\mu$ is only sensible if the file addresses $\add(C)$ are approximately uniformly distributed. Otherwise, $\mu(A)$ should correspond to the measure of the preimage $\add^{-1}(A)$.
However, on the one hand, the latter definition requires an introduction to measure theory.
On the other hand, we would need to express the difference between a pseudo-random hash function, which we use for constructing the addressing scheme $\add$, and a uniform distribution in terms of the parameter $\delta$, which is out-of-scope for this paper.
Thus, we restrict our goals to balancing the fraction of content addresses mapped to a node rather than the fraction of content mapped to the node.

\subsection{Dynamics and Stabilization}
Now, we assume that the topology of the graph changes over time.
Here, each topology change refers to the addition and removals of \emph{one} node or edge.
We model the graph topology over time as a stochastic process $(G_t)_{t \in \mathbb{N}}$ such that $G_t=(V_t,E_t)$ represents the graph after the $t$-th topology change.
When the topology changes, the embedding has to be adapted, so that we have a time-dependent embedding $id_t$. 
In contrast, we assume that the set of potential content $C$ and the addressing function $\add$ remain unchanged.
In order for the content addressable storage $(G_t,\id_t, C, \add)$ to continuously function effectively for all $t$, the embedding has to be adjusted to maintain greedy and balanced. We now define two properties for an embedding algorithm, before formally defining the concept of a content addressable storage in a dynamic scenario.

\begin{definition}
\label{def:dymEm}

\stefout{Let $\embed$ be an embedding algorithm generating content addressable greedy embeddings $\id$ based on spanning tree $T = (V, E^T), E^T \subseteq E$. 
Assume that we replace the subtree $T_u=(V_u, E^T_u)$ of $T$ rooted at $u \in V$ by a subtree $T'_u=(V'_u,E'_u)$ with the same root. 
The algorithm $\embed$ is called \emph{dynamic} if $\embed(T'_u,\id(u))$ generates content addressable greedy embeddings $\id'$ such that
}
{Let $\embed$ be an embedding algorithm for content addressable greedy embeddings $\id$ based on spanning trees $T = (V, E^T), E^T \subseteq E$.
We write $\embed(T,\emptyset)$ to indicate that $\embed$ is applied on the tree $T$. 
We call $\embed$ dynamic if we can compute $\embed(T_u, \id(u))$ on a subtree $T_u=(V_u, E^T_u)$ with $V_u \subseteq V, E^T_u \subseteq E^T$ rooted at a node $u$ such that
\begin{enumerate}
\item $\embed(T_u, \id(u))$ only changes coordinates of nodes $v \in V_u$,
\item the communication complexity of $\embed(T_u, \id(u))$ is $\calO(|V_u|)$, and
\item for any tree $T'_u = (V'_u, E^{T'}_u)$ rooted at $u$, graph $G'=(V', E')$ with $V'=V\setminus V_u \cup V'_u$ and spanning tree $T'=(V',E^{T'})$ with $E^{T'}=E^T\setminus E^T_u \cup E^{T'}_u$, $\embed(T'_u, \id(u))$ results in an embedding $\id'$ such that 
$$\forall x \in \X', \m(x) \in V_u \implies \m'(x) \in V'_u.$$
\end{enumerate}
Furthermore, $\embed$ is called \emph{dynamic $(f',\delta)$-balanced} if it is dynamic and
\begin{align}
\label{eq:dymCA}
\forall v \in V_u: \mu(\closest(v)) \leq \left(\sum_{v_0 \in V_u}\mu(\closest(v_0))\right)\frac{f'}{|V_u|}+\delta. 
\end{align}
holds for any embedding generated by $\embed(T_u, \id(u))$.  
}
\end{definition}

In other words, Definition \ref{def:dymEm} requires an embedding algorithm to be able to re-embed local subtrees with changed nodes and edges such that the local embedding is balanced, covers the addresses of the previous embedding, and other nodes and their content addresses are unaffected.
Note that this local balance does not imply global balance.
If the combined fraction of coordinates $M_0=\sum_{v_0 \in V_u}\mu(\closest(v_0))$ mapped to the nodes in the subtree is disproportionally high in comparison to the number of nodes in the subtree, the fraction of coordinates mapped to each node in the subtree might exceed $f/n$. 
A stabilization algorithm should thus decide if the embedding algorithm $\embed$ can be applied locally or if the re-embedding has to consider additional nodes in order to balance the storage responsibilities.

\begin{definition}
\label{def:stab}
Let $\embed$ be a $(f', \delta)$-balanced embedding algorithm with $f' \leq f$.
A stochastic process $((G_t,\id_t)_{t \in \mathbb{N}}, C, \add, \stab(\embed))$ is called a \emph{dynamic $(f, \delta)$-balanced content addressable storage} if the \emph{stabilization algorithm} $\stab(\embed)$ ensures that 
$(G_t,\id_t, C, \add)$ is a $(f,\delta)$-balanced content addressable storage for all $t \in \mathbb{N}$. 
\end{definition}
Definition \ref{def:stab} allows the stabilization algorithm to be parameterized by the embedding algorithm.
In this manner, we allow for a general stabilization algorithm that calls upon a variable dynamic embedding algorithm.

%% file: algorithm.tex
\section{Algorithm Design}
\label{sec:algo}

In this section, we develop an efficient stabilization algorithm that can restore a $(\calO(D),\delta)$-balanced content addressable storage after a topology change with $D$ denoting an upper bound on the spanning tree depth. 
We first consider the algorithm design from a high-level point of view.
More precisely, we show that we can construct such a stabilization algorithm $\stab(\embed)$ on the basis of any dynamic $(1,\delta)$-balanced embedding algorithm $\embed$.
We then present a concrete algorithm $\embed$ for our evaluation.
Last, we introduce potential variations and improvements of our algorithm for practical use. 

\subsection{$Stabilization$}

The key idea of algorithm $\stab(\embed)$ is that a node $u$ can locally decide if it re-embeds its subtree or forwards a request for re-embedding to its parent. 
Throughout this section, let $T_u=(V_u, E_u)$ denote the subtree rooted at $u$.
In order to decide if a local re-embedding is possible, $u$ maintains an estimate $n_{est} \in [n/g,ng]$ of the number of nodes $n$ in the network.
Furthermore, $u$ keeps track of its number of descendants $|V_u|$ as well as the combined fraction $cont(V_u)=\sum_{v \in V_u}\mu(\closest(v))$ of addresses assigned to nodes in $V_u$.  
Similarly, for each child $v$, $u$ keeps track of $|V_v|$. 
We first describe the idea of how the dynamic re-embedding in the presence of topology changes works. 
Then, we detail how to obtain the required knowledge for making the local decision to re-embed. \longonly{Last, we present the pseudocode of the stabilization algorithm.}  

\paragraph{Maintaining Stability} We aim to maintain a $(f, \delta)$-balanced content addressable storage with $f=\calO(D)$ in the presence of topology changes. We assume that there exists a $(1, \delta)$-balanced embedding algorithm $\embed$. 
The topology change and subsequent spanning tree stabilization either replaces a subtree $T_u$ with a subtree $T'_u=(V'_u, E'_u)$ or construct a new spanning tree. We focus on the former case as the latter requires re-embedding the complete graph. 
If $u$ re-embeds locally, i.e., applies the algorithm $\embed$ only to $T'_u$, we have $cont(V'_u)=cont(V_u)$ by the third condition in Definition \ref{def:dymEm}.     
Then Eq. \ref{eq:dymCA} states that the maximal fraction of addresses assigned to any node $v$ in $V'_u$ is
\begin{align}
\label{eq:clv}
\mu(\closest'(v)) \leq \frac{cont(V_u)}{|V'_u|} + \delta,
\end{align}
because $\embed$ is $(1,\delta)$-balanced. 
We can express Eq. \ref{eq:clv} in the form $\frac{f'}{n}+\delta$ with $f'=n \cdot cont(V_u)/|V'_u|$.
If indeed $n_{est} \in [n/g,ng]$ for a global parameter $g$, we have $n \leq n_{est}g$ and hence
$f' \leq n_{est}g \frac{cont(V_u)}{|V'_u|}$.
Thus, if for some $f(u)\leq f$
\begin{align}
\label{eq:criterion}
n_{est}g \frac{cont(V_u)}{|V'_u|} \leq f(u), 
\end{align}
re-embedding locally guarantees that $\mu(\closest(v)) \leq \frac{f}{|V|}+\delta$ for all $v \in V'_u$, so that we indeed maintain a $(f, \delta)$-balanced content addressable storage.  
If Eq. \ref{eq:criterion} does not hold, $u$ cannot guarantee that local re-embedding maintains a $(f(u), \delta)$-balanced content addressable storage. Then $u$ contacts its parent $p(u)$ with a request for re-embedding. The node $p(u)$ decides if it should re-embed locally, changing the coordinates within subtrees rooted at $u$ and its siblings, or if it relays the request to its own parent.
In this manner, nodes might forward the request for re-embedding to the root who can always re-embed such that the resulting content addressable storage is $(1, \delta)$- and hence $(f , \delta)$-balanced.

It remains to consider how to choose $f(u)$. 
As stated above, we need to ensure that $f(u)\leq f$. On the first glance, the choice $f(u)=f$ seems suitable as it maximizes the probability that Eq. \ref{eq:criterion} holds. However, if indeed $f=f(u)$, the re-embedding might only barely restore the desired balance. Any further change affecting any of the subtrees might thus lead to an immediate need for another re-embedding. In order to allow to reduce the frequency of the re-embedding, we thus choose a level-dependent $f(u)$. More precisely, a parent $v$ provides an embedding with a lower balance factor than the child node $u$, i.e.,  $f(v)< f(u)$. In this manner, the probability that $u$ has to contact its parent for a further re-embedding shortly after such an re-embedding decreases.
In order to maintain an overall balance factor $f=\calO(D)$, we choose the local balance factor corresponding to the level of the node in the tree, i.e.,
\begin{align}
\label{eq:fu}
f(u) = g(1+c+level(u)).
\end{align}
Using the size approximation accuracy $g$ as factor ensures that a re-embedding is not only necessary due to the uncertainty about the current global state despite a good balance in the subtree. The \emph{tree depth offset} $c$ allows a trade-off between the accepted level of imbalance and the stabilization overhead. So, an increased parameter $c$ implies that the maximal fraction of content per node can be high but might reduce the frequency of coordinate changes.

\paragraph{Updating State Information} The estimate $n_{est}$ as well as the quantities $|V_v|$ and $cont(V_u)$ are essential to check if Eq. \ref{eq:criterion} holds.
%
The fraction $cont(V_u)$ depends on the nature of the coordinate space $\X$ and the address space $\X' \subset \X$. Thus, computing them depends on the nature of the embedding algorithm $\embed$ and the addressing scheme $\add$. Here, we thus only describe how to obtain $cont(V_u)$ during the design of $\embed$. 
Here, we focus on maintaining the network size estimate $n_{est}$. In the process, we also obtain and maintain the subtree sizes $|V_v|$.
Note that $n=|V_r|$ for the root $r$.
Upon initialization, we derive the network size $n=|V|=|V_r|$ recursively. Each node $v$ forwards the size $V_v$ to its parent, starting at leaves, which send $|V_v|=1$. As soon as a node $u$ has received $|V_v|$ from all its children $v$, $u$ sends 
$1+\sum_{v\in children(u)} |V_v|$ to its parent. 
Finally, $r$ obtains the current network size and broadcasts it to all nodes along the edges of the tree. 
Later on, whenever a node $u$ accepts an additional child or one of its children departs, $u$ sends the new value of $|V_u|$ to its parent.
All subtree sizes along the path to the root are subsequently updated. 
After the root node has updated its local state, it checks if the current value for $|V_r|= n$ and the global estimate $n_{est}$ still satisfy 
$n_{est} \in [n/g, ng]$. If not, $r$ broadcasts the new estimate and at the same time runs the re-embedding algorithm.
\shortonly{Algorithm \ref{algo:stabembed} describes a node $u$'s reaction  to a topology change in $T_u$, covering both the network size estimation maintenance and the embedding, as exemplarily illustrated in Fig. \ref{fig:sa_example}.
}

\begin{algorithm}
\caption{$\stab(\embed)(u,v, |V'_v|,b)$}
\label{algo:stabembed}
\begin{algorithmic}[1]
{\small
\STATE \COMMENT{$u$: node, $v$: child of $u$; $|V'_v|$: updated $|V_v|$, $b$: forward flag}
\STATE \COMMENT{Global: balance factors $f$, $g$, $c$; size estimate $n_{est}$; Alg. $\embed$}
\STATE \COMMENT{State $u$: content $cont(V_u)$; subtree sizes $|V_v|$; parent $p(u)$}
\STATE $|V_v| = |V'_v|$ \label{l:Vv} 
\STATE $|V_u| = 1 + \sum_{v\in children(u)}|V_v|$ \label{l:Vu}
\IF{$u$ is root} \label{l:rootS}
  \IF{$|V_u| < n_{est}/g$ or $|V_u| > n_{est}g$ or not $b$}   
    \STATE $n_{est}=|V_u|$    
    \STATE $\embed(u)$,
    \STATE Broadcast $n_{est}$ \label{l:rootE}
  \ENDIF  
\ELSIF{$b$}
   \STATE $\stab(\embed)(p(u),u, |V_u|,b)$ \label{l:forward}
\ELSE
   \IF{$n_{est}g \frac{cont(V_u)}{|V_u|} \leq g(1+c+level(u))$} \label{l:check}
  \STATE $\embed(u)$ \label{l:em1}
  \STATE $\stab(\embed)(p(u),u, |V_u|,true)$ \label{l:em2}
 \ELSE
  \STATE $\stab(\embed)(p(u),u, |V_u|,false)$  \label{l:relay}
 \ENDIF
\ENDIF
}
\end{algorithmic}
\end{algorithm}  

\longonly{\paragraph{Pseudocode} Algorithm \ref{algo:stabembed} displays the pseudo code governing a node $u$'s reaction  to a topology change in $T_u$.
The algorithm combines the decision for re-embedding with updates of local state information.
The input of the algorithm is the current node $u$, the child $v$ that is affected by the change, the new value for $|V_v|$, and
a flag $b$ indicating that the re-embedding has already been taken care of. Hence, if $b$ is true, $u$ only has to forward the updated
subtree sizes to the root.  
The system parameters are the balance factor $f$ and the estimation quality $g$. In addition, each node stores the same global network size estimate $n_{est}$. The local state at $u$ includes the fraction of content $cont(V_u)$, the number of nodes in subtrees rooted at its children, and the parent $p(u)$. 
In Lines \ref{l:Vv} and \ref{l:Vu}, $u$ updates its information regarding the subtree sizes. 
Lines \ref{l:rootS}-\ref{l:rootE} specify the behavior of the root.
Note that the root can always generate an $(f,\delta)$-balanced content addressable storage, so that checking Eq. \ref{eq:criterion} are not necessary.  Rather, the root calculates a new estimate $n_{est}$ and re-embeds the graph whenever the old estimate is not accurate enough or its descendants have been unable to locally re-embed, i.e., if $b$ is false.
If $u$ is not the root of the tree, $u$ first checks the flag $b$. If $b$ is true, $u$ merely relays the updated subtree sizes to the parent (Line \ref{l:forward}).
Otherwise, $u$ has to decide if it locally re-embeds or relays the request for re-embedding to its parent. 
The decision in Line \ref{l:check} follows Eq. \ref{eq:criterion}. 
If the local state information satisfies Eq. \ref{eq:criterion}, $u$ executes the embedding algorithm on $T_u$ and forwards the updated subtree sizes to its parent. Furthermore, $u$ sets the flag $b$ to true (Lines \ref{l:em1} and \ref{l:em2}).
If $u$ cannot maintain the necessary balance by locally re-embedding, $u$ forwards the updated subtree sizes to the parent together with $b$ set to false, indicating the need for a re-embedding (Line \ref{l:relay}).
In this manner, Algorithm \ref{algo:stabembed} recursively restores a $(f,\delta)$-balanced content addressable storage.}

\begin{figure*}[!t]

   \centering
   \subfloat[Perfect initial balance under $v_1$.]{
      \includegraphics[width=0.24\textwidth]{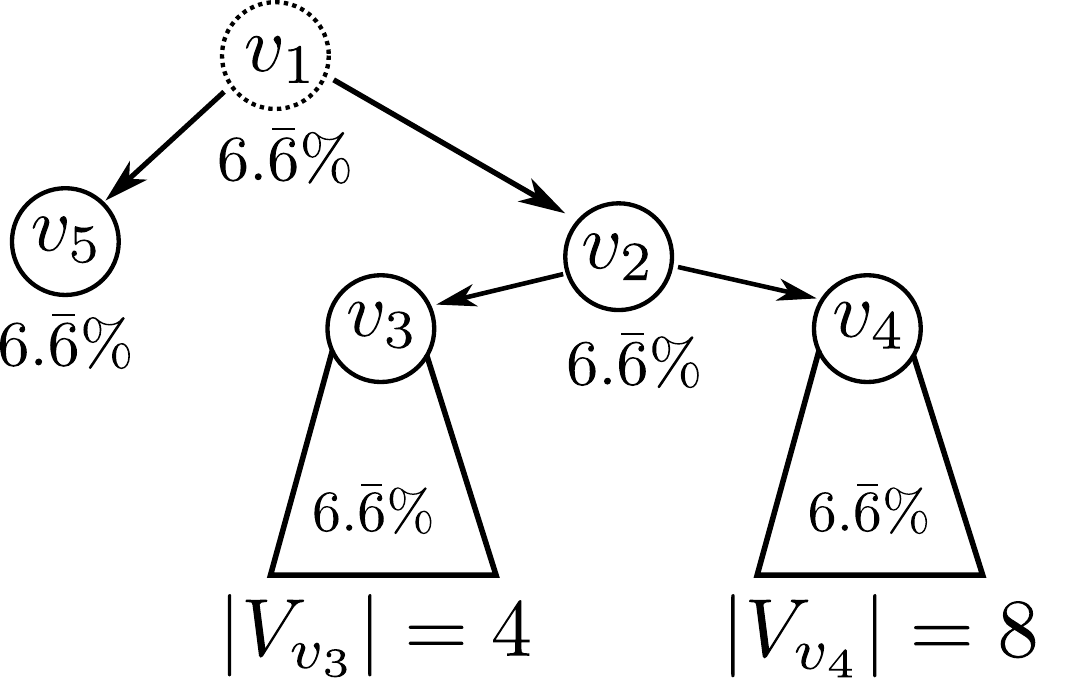}
      \label{fig:sa_example_a}
   }
   \hfill
   \subfloat[Upon departure of 5 children of $v_4$.]{
      \includegraphics[width=0.24\textwidth]{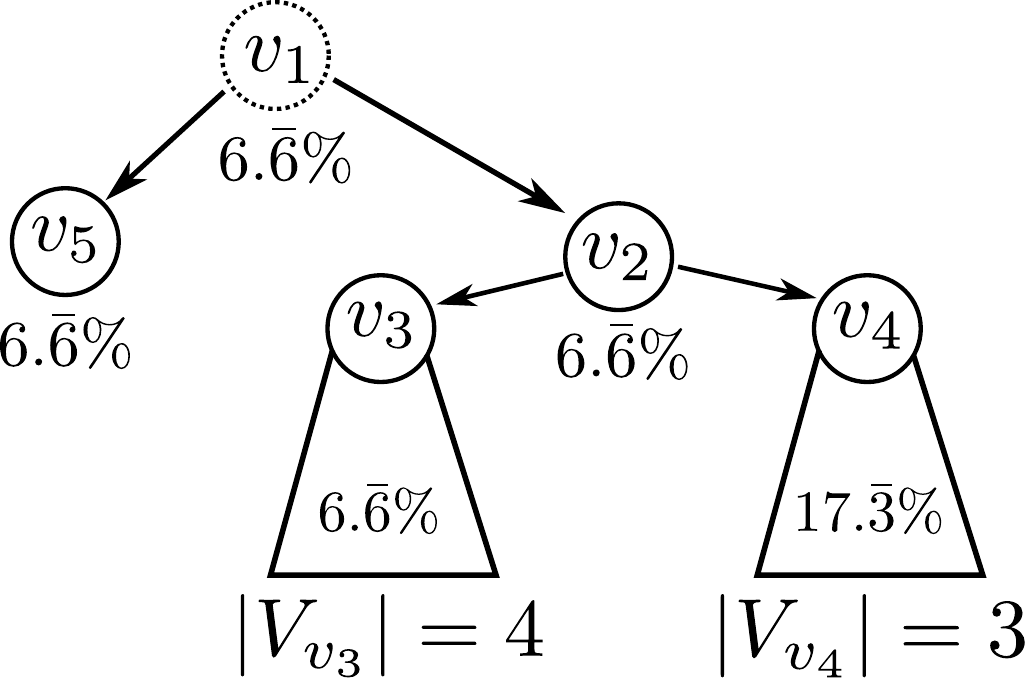}
      \label{fig:sa_example_b}
   }
   \hfill
   \subfloat[Re-balancing on $6^{th}$ departure at $v_4$.]{
      \includegraphics[width=0.24\textwidth]{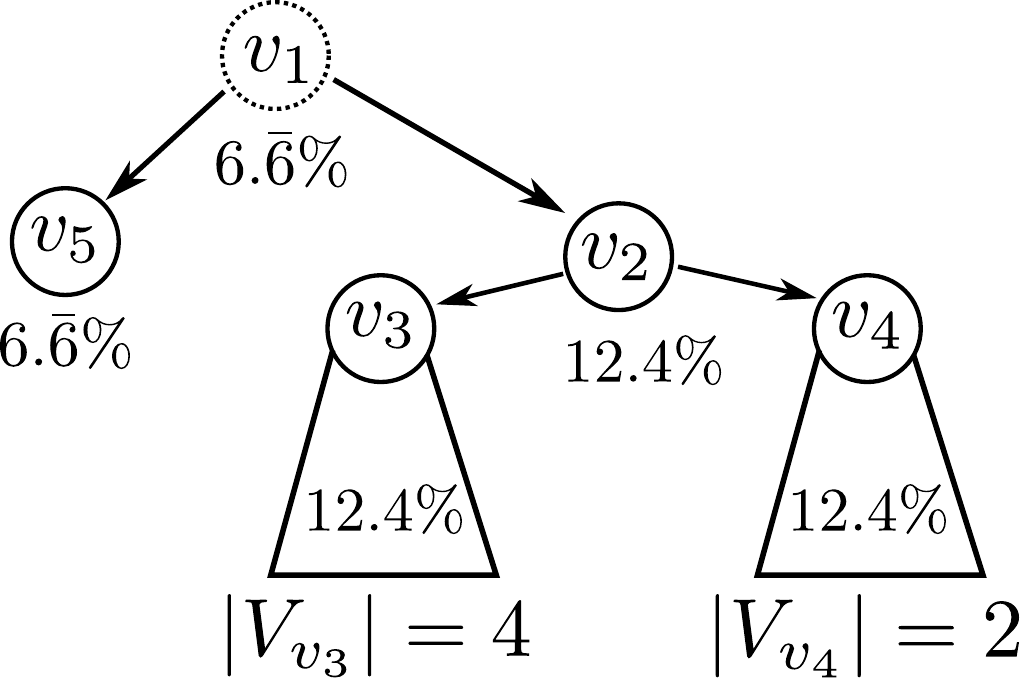}
      \label{fig:sa_example_c}
   }
   \caption{$\stab(\embed)$ on 15 nodes (triangles below $v_3$ and $v_4$ denote branches of $|V_{v3}|$ and $|V_{v4}|$ nodes), $g = 2$ and $c = 0$:
Arrows denote parent child relationships.
Percentages given for branches denote the fraction allocated to each node in the branch, percentage given for top nodes the fraction allocated to the respective node; 
Fig. \ref{fig:sa_example_a}: the content addressing initially allocates $\nicefrac{1}{15}^{th}$ to each node, achieving perfect balance.
Departure of 5 children under $v_4$ is balanced by $v_4$ (Fig. \ref{fig:sa_example_b}), allocations increase to $17.3\%$ per node.
The sixth departure triggers escalation and the re-embedding request is relayed to $v_2$ (Fig. \ref{fig:sa_example_c}).
}
   \label{fig:sa_example}
   \vspace{-1em}
\end{figure*}

We consider potential variations and speedups for the use of $\stab(\embed)$ in practice in Section \ref{sec:vary}. 

\subsection{Embedding}
Our embedding algorithm is a modification of the unbalanced content addressing scheme for Prefix Embedding \cite{hofer2013greedy}. 
In a nutshell, the idea of our algorithm is to count each node in Prefix Embedding as multiple nodes. 

Prefix Embedding, a variation of PIE \cite{HerzenEtAl11}, encodes the position of a node with regard to the root of the tree.
More precisely, each node $u$ enumerates the edges to its children. 
The coordinate then corresponds to the vector of edge numbers on the unique path from the root to the respective node.
The distance between two such coordinates corresponds to the sum of the length of their coordinate vectors, subtracting twice the number of leading equal elements (the common prefix).
In this manner, the distance between two node coordinates equals the length of the path between the two nodes in the spanning tree.

We modify Prefix Embedding by replacing the numerical elements of the vectors with sets of integers in an interval. The length of each interval depends on the number of nodes in the corresponding subtree. 
So, a node on level $l$ divides the space of $2^b$ numbers for the $l+1$-th element of the coordinate vector evenly between itself and its descendants, as displayed in Algorithm \ref{algo:embed}.
More precisely, a node $u$ receives a prefix for all nodes in $V_u$ from its parent, starting with an empty prefix at the root. 
The prefix corresponds to $u$'s own coordinate $\id(u)$ and consists of $l$ intervals. 
After receiving its coordinate, $u$ assigns coordinates consisting of $l+1$ intervals to its children. 
For the $i$-th child $v_i$, $u$ adds the set of integers in the interval 
$\left[\sum_{j=1}^{i-1}\frac{|V_{v_j}|}{|V_u|}2^b,\sum_{j=1}^{i}\frac{|V_{v_j}|}{|V_u|}2^b\right]$, which has cardinality of at most 
$\lceil \frac{|V_{v_i}|}{|V_u|}2^b \rceil$ (Lines \ref{l:childS}-\ref{l:childE}).  The subtree rooted at the child is then recursively embedded. 

We now formally derive the coordinate space and the distance function for the assigned coordinates. 
First, we denote the set of all integers within an interval 
$[z_1,z_2)$ by $\integers(z_1, z_2) = \{ i : i \in [z_1,z_2), i \in \mathbb{Z}\}$. Furthermore, let $IC= \{\integers(z_1, z_2): z_1, z_2\in [0,2^b), z_2 \geq z_1\}$ denote the set of all such sets with 
$0 \leq z_1 \leq z_2 < 2^b$. Then, our coordinate space $\X=IC^*$ corresponds to all vectors with entries in $IC$.
The distance between two node coordinates is analogous to Prefix Embedding: the difference of the sum of the coordinates lengths and twice the common prefix length. However, in order to allow for balanced content addressing, we use a slightly different concept than the common prefix length to compare vector elements.
Rather than only considering equal elements, we consider two sets a match if one is contained in the other. 
Formally, let $I_1$ and $I_2$ denote two sets of integers. Then we set $sub(I_1,I_2)=true$ if $I_1\subseteq I_2$ or $I_2 \subseteq I_1$.
As a consequence, we denote the \emph{contained interval length} of two vectors $x_1, x_2 \in \X$ as
$cil(x_1,x_2)=\max \{j \in \{0, \ldots , \min\{D(x_1), D(x_2)\}\}: sub(x_1(j),x_2(j))\}$.
Hence, the distance between two coordinates is 
\begin{align}
\label{eq:d}
\dist(x_1,x_2)= D(x_1)+D(x_2)-2\cdot cil(x_1, x_2).
\end{align}
with $D(x)$ denoting the dimension of a vector $x$.  

Next, we consider the file address generation $\add$ for files $c \in C$.
Typically, $\add$ corresponds to a hash function $h: C \rightarrow H=\mathbb{Z}_{2^b}$. 
However, as our coordinates are vectors, we choose the address space $\X'=\{\{a\}: a \in \mathbb{Z}_{2^b}\}^L$ corresponding to vectors of a fixed length $L$ with $L\geq D$ exceeding the spanning tree depth. 
We use multiple salted hashes to obtain the address of a file $c$, i.e., the $i$-th element of $\add(c)$ is $y=(\{y_1\}, \ldots \{y_L\})$ for $y_i=h(c + i)$.
The file $c$ is then stored at the node with the closest coordinate $\id(u)$ to $\add(c)$ according to Eq. \ref{eq:d}.
The node $u$ can be located using greedy routing by forwarding a request to store or retrieve $\add(c)$ to the closest neighbor until no such neighbor exists. 
If the coordinate of $u$ changes due to the dynamics, $u$ needs to start a new storage request for $c$ to ensure that the file is indeed stored on the closest node. 
We now show that $\embed$ is greedy content addressable, $(1,\delta)$-balanced, and hence any $(G,\id, C, \add)$ is a content addressable storage.

\begin{algorithm}
\caption{$\embed(T_u, \id(u))$}
\label{algo:embed}
\begin{algorithmic}[1]
{\small
\STATE \COMMENT{$T_u$:subtree to embed, $Ic(z_1,z_2)$: integers in $[z_1,z_2)$}
\STATE \COMMENT{$b$: length of coordinate elements, $||$: concatenation}
\STATE $ol=0$
\STATE $next=0$
\FOR{$v \in children(u)$} \label{l:childS}
\STATE $next=next+|V_v|$
\STATE $\id(v)=id(u)||\integers(\frac{ol}{|V_u|}2^b , \frac{next}{|V_u|}2^b)$ \label{l:interval}
\STATE $\embed(T_v, \id(v))$
\STATE $ol=next$\label{l:childE}
\ENDFOR
}
\end{algorithmic}
\end{algorithm}

\begin{prop}
The dynamic embedding algorithm $\embed$ is content addressable greedy. 
If $L \leq b$ with $L$ being the upper bound on the tree depth, then $\embed$ is $(1,\frac{L+1}{2^b})$-balanced.
\end{prop}


\begin{proof}
\shortonly{[Sketch, we omit the technical details of some steps and focus on the key ideas, see \cite{longvers} for a detailed proof]}

In order to show that $\embed$ is content addressable greedy, we leverage the corresponding results for Prefix Embedding.
In the main part of the proof, we show that $\embed$ is $(1,\frac{L+1}{2^b})$-balanced by determining $f$ and $\delta$ as in Definition \ref{def:fbalance}. For this purpose, we first determine an upper bound for $\closest(v)$ based on Eq. \ref{eq:d}
and then leverage this upper bound to confirm that $\delta \leq \frac{L+1}{2^b}$ for $f=1$.

\longonly{At first, we show that the coordinate spaces $\X$, $\X'$ as well as the distance $\dist$ defined in Eq. \ref{eq:d} can be mapped to a Prefix Embedding. Before showing that $\embed$ is content addressable greedy, we present an alternative interpretation of the coordinate spaces $\X$ and $\X'$ and the distance $\dist$ defined in Eq. \ref{eq:d}.
It is easy to see that $\X'$ and $\mathbb{Z}^L_{2^b}$ are equivalent: we map a vector of singletons $x'=(\{a_1\}, \ldots , \{a_l\})$ to a vector $\tilde{x}'=(a_1,\ldots , a_l)$.
Similarly, we can associate $\X$ with a subset $Y$ of $\mathcal{P}\left(\mathbb{Z}_{2^b}^*\right)$ with $\mathcal{P}$ indicating the power set, i.e., the set of all subsets.
Here, we map a vector $x=(Ic(x_1(1),x_1(2)), \ldots , Ic(x_l(1),x_l(2)))$ to the set $\tilde{x}=\{(y_1, \ldots y_l): y_i \in Ic(x_i(1),x_i(2))\}$. 
Now, $\{a_i\} \subseteq Ic(x_i(1),x_i(2))$ holds iff there exists $y \in \tilde{x}$ with $y_i=a_i$. 
For $c=\max_{y \in \tilde{x}} cpl(y,\tilde{x}')$ and by the definition of $cil(x, x')$, we have 
$cil(x,x')\leq c$ and $cil(x,x')\geq c$, so that $cil(x,x')=\max_{y \in \tilde{x}} cpl(y,\tilde{x}')$. 

Leveraging the alternative definition of $cil$, we show that $\embed$ is content addressable greedy by relating $\embed$ to Prefix Embedding. Consider a node $u$ with coordinate $\id(u)=\left(Ic(x^u_1(1),x^u_1(2)), \ldots , Ic(x^u_l(1),x^u_l(2))\right)$. 
Replace $u$ with a set of nodes $ID(u)$ of size $\prod_{i=1}^l|Ic(x^u_i(1),x^u_i(2))|$ and assign each node $u' \in ID(u)$ a unique coordinate $\id_{PRE}(u')=(y_1, \ldots ,y_l)$ with $y_i \in Ic(x^u_i(1),x^u_i(2))$.
For every neighbor $v$ of $u$, connect $u'$ with all $v' \in ID(v)$. 
In particular, we have an edge between $u'$ and all $v' \in ID(v)$ such that $\id_{PRE}(v')=(y_1, \ldots ,y_l,z)$ for some integer $z$.
Thus, the resulting embedding $\id_{PRE}$ is an instance of Prefix Embedding and hence content addressable greedy as shown in \cite{hofer2013greedy}.
As a consequence, greedy routing for an address $\tilde{x}' \in \mathbb{Z}^L_{2^b}$ traverses a path $(v'_0, \ldots , v'_t)$ such that $\id_{PRE}(v'_t)$ is closest to $\tilde{x}'$ in terms of $\distPre(\tilde{x}',\id_{PRE}(v'_t))=D(\tilde{x}')+D(\id_{PRE}(v'_t))-2cpl(\tilde{x}',\id_{PRE}(v'_t))$.
For the equivalent address $x' \in \X'$, we have $cil(x',\id(v))=\max_{v' \in ID(v)}cpl(\tilde{x}',\id_{PRE}(v'))$ and $\dist(x',\id(v))=\min_{v' \in ID(v)}\distPre(\tilde{x}',\id_{PRE}(v'))$.
Hence, greedy routing for an address in the embedding $\id$ traverses the path  $(v_0, \ldots , v_t)$ and terminates at the node with the closest coordinate to $x'$.
Thus, $\embed$ is content addressable greedy.}

\shortonly{Note that $\embed$ corresponds to a Prefix Embedding when we represent each node $u$ as a set of nodes $U$ with integer-valued vectors as coordinates. The coordinate of a node $u' \in U$ is chosen such that the $i$-th element is contained in $Ic(x^u_i(1),x^u_i(2))$ and $U$ is the union of all nodes with such coordinates. As PrefixEmbedding is content addressable greedy \cite{hofer2013greedy}, $\embed$ is content addressable greedy.} 

Now, we show that $\embed$ is $(1,\frac{L}{2^b})$-balanced.
\longonly{First, we derive the cardinality $|\closest(v)|=\mu(\closest(v))\cdot (2^b)^L$. Let $x=id(v)$.
By the definition of the distance $\dist$ in Eq. \ref{eq:d}, $\closest(v)$ consists of all vectors $x' \in X'$ such that i) $cil(x,x')=D(x)$, and ii) there is no child $u$ of $v$ with $cil(\id(u),x')>D(x)$.
We extend $x$ to a coordinate $x_L$ of length $L$ such that $\closest(v)=\{x'\in \X': \dist(x',x_L)=0\}$ by adding elements $Ic(x^v_{i}(1),x^v_{i}(2))$ for $i>D(x)$.
With
\begin{align*}
z_{m}(v) &= \begin{cases}
0, &children(v)=\emptyset \\
\max_{w \in children(v)}x^w_{D(x)+1}(2), &\textnormal{otherwise}
\end{cases}
\end{align*}
   denoting the smallest integer such that $z_{m} \notin Ic(x^w_{D(x)+1}(1),x^w_{D(x)+1}(2))$ for any child $w$ of $v$, we set $Ic(x^v_{D(x)+1}(1),x^v_{D(x)+1}(2))=Ic(z_{m}(v), 2^b)$. Furthermore, for $i > D(x)+1$, we set $Ic(x^v_{i}(1),x^v_{i}(2))=Ic(0,2^b)$ to cover all possible elements. On the one hand, $x'=(\{a_1\}, \ldots , \{a_L\}) \in \closest(v)$ implies $a_i \in Ic(x^v_i(1), x^v_i(2))$ for all $i$ and hence $\dist(x_L,x')=0$.
On the other hand, $x' \notin \closest(v)$ implies $a_i \notin Ic(x^v_i(1), x^v_i(2))$ for some $i\leq D(x)+1$ and hence $\dist(x_L,x')\neq 0$.
So, indeed $\closest(v)=\{x'\in \X': \dist(x',x_L)=0\}$ and hence
\begin{align}
\label{eq:bva}
\begin{split}
|\closest(v)| &= \prod_{i=1}^{L}|Ic(x^v_i(1), x^v_i(2))| \\
              &=2^{b(L-D(x)-1)}\prod_{i=1}^{D(x)+1}|Ic(x^v_i(1), x^v_i(2))|.
\end{split}              
\end{align}

In the following, let $v_l$ denote the ancestor of $v$ on level $l$. In particular, $v=v_{D(x)}$. 
   As stated in Line \ref{l:interval} of Algorithm \ref{algo:embed}, we have $Ic(x^v_i(1), x^v_i(2)) = \integers(\frac{r}{|V_{v_{i-1}}|}2^b, \frac{r + |V_{v_i}|}{|V_{v_{i-1}}|}2^b)$ for $i \leq D(x)$ with some $r \in [0, |V_{v_{i-1}}| - |V_{v_{i}|})$. This implies an upper bound $|Ic(x^v_i(1), x^v_i(2))| \leq \lceil \frac{|V_{v_i}|}{|V_{v_{i-1}}|} 2^b \rceil \leq \frac{|V_{v_i}|}{|V_{v_{i-1}}|}2^b + 1$. For $i=D(x)+1$, the number of integers not assigned to any of the children is bound by
$\frac{1}{|V_v|} 2^b +1$. Inserting these bounds into Eq. \ref{eq:bva} yields
\begin{align}
\label{eq:bv}
|\closest(v)| \leq 2^{b(L-D(x)-1)} \left(\frac{1}{|V_v|} 2^b +1\right) \prod_{i=1}^{D(x)}\left(\frac{|V_{v_i}|}{|V_{v_{i-1}}|} 2^b + 1\right).
\end{align}

In the second step of the proof, we derive $\delta$ from Eq. \ref{eq:bv}.
For this purpose, we write 
\begin{align}
\label{eq:ci}
\prod_{i=1}^{D(x)}\left(\frac{|V_{v_i}|}{|V_{v_{i-1}}|}2^b + 1\right) = \sum_{i=0}^{D(x)} c_i (2^b)^i
\end{align}
with 
\begin{align*}
c_{D(x)}=\prod_{i=1}^{D(x)}\frac{|V_{v_i}|}{|V_{v_{i-1}}|}=\frac{|V_{v_{D(x)}}|}{|V_{v_{0}}|}=\frac{|V_{v}|}{n}. 
\end{align*}
As $\frac{|V_{v_i}|}{|V_{v_{i-1}}|}\leq 1$, we have $\prod_{i=1}^{D(x)}\left(\frac{|V_{v_i}|}{|V_{v_{i-1}}|}2^b + 1\right) \leq \left( 2^b + 1 \right)^{D(x)}=\sum_{i=0}^{D(x)} {D(x) \choose i} (2^b)^i$ and $c_i \leq {D(x) \choose i}$.
Consequently, we obtain upper bounds $c_{D_{x-1}}\leq D(x)\leq L $ and for $D(x)>1$
\begin{align*}
&\sum_{i=0}^{D(x)-2} c_i (2^b)^i \leq (2^b)^{D(x)-2} \sum_{i=0}^{D(x)-2} {D(x) \choose i} \\
&\leq (2^b)^{D(x)-2} \sum_{i=0}^{D(x)} {D(x) \choose i} = (2^b)^{D(x)-2} 2^{D(x)}\leq (2^b)^{D(x)-1}
\end{align*}
The last step uses $2^b \geq 2^L \geq 2^{D(x)}$.
Inserting Eq. \ref{eq:ci} and the upper bounds on the coefficients $c_i$ in Eq. \ref{eq:bv}, we obtain
\begin{align*}
|\closest(v)| \leq \frac{1}{n}\left(2^b\right)^L + L\left(2^b\right)^{L-1} + \left(2^b\right)^{L-1}.
\end{align*} 
Division by the number of addresses $|X'|=(2^b)^L$ shows that $\embed$ is indeed $(1,\frac{L+1}{2^b})$-balanced. 
}

\end{proof}

\subsection{Improvements and Variations}
\label{sec:vary}
There are multiple possibilities to slightly reduce the overhead of $\stab(\embed)$ in practice or achieve additional properties.
\shortonly{We discuss additional variants regarding the use of estimated subtree sizes and the adherence to heterogeneous node resources in our technical report \cite{longvers}\footnotemark[1]}. 

\longonly{
\textit{Delay before broadcasting new estimate:}
In Line \ref{l:rootE} of Algorithm \ref{algo:stabembed}, the root broadcasts a new estimate immediately after receiving an update.
However, a node or edge departure might result in a temporarily low estimate as the descendants of the departed nodes select alternative parents. Adding a short delay before reacting to a considerable change in the network size avoids broadcasting a new estimate without the actual need to do so. 
During our theoretical and practical evaluation, we assume that the root waits until its size estimate is accurate.  
}

\textit{No local re-embedding after joins:}
In Algorithm \ref{algo:stabembed}, the parent $u$ of a newly joined node $v$ re-embeds the complete subtree $V_v$. However, as long as the load is sufficiently balanced in the subtree, such an action might unnecessarily increase the overhead. Rather, $u$ can assign $v$ a preliminary coordinate in $\closest(u)$ by adding $Ic(max, (max+2^b)/2)$ to $v$'s coordinate with $max$ denoting the highest number assigned to a last coordinate of $u$'s children. 
In this manner, we postpone the re-embedding of the subtree until one of its children leaves or asks for a re-embedding.

\textit{Virtual binary trees:}
In addition to reducing the frequency of new coordinate assignments, the number of nodes affected by a re-embedding can be reduced by only
changing the coordinates within a subset of the subtrees rooted at children. For this purpose, we leverage the concept of virtual binary trees presented in \cite{hofer2013greedy}. Here, we represent a subgraph consisting of a parent and its children as a binary tree such that the children are the leaves and the parent executes the functionality of all internal nodes. 
In this manner, if a node $u$ receives a re-embedding request relayed from one of its children, $u$ first checks if first checks if it can balance a set of two or three subtrees. If re-embedding only those trees is possible according to Algorithm \ref{algo:stabembed}, the remaining subtrees remain unaffected. 
Otherwise, nodes subsequently considers subtrees at a lower level of the virtual binary tree until it can either locally re-embed or has to relay the request to its own parent.
As the structure of the virtual subtree rooted at a node changes whenever the number of children changes, the successive consideration only applied for nodes other than the parent of the joined or departed node.  
Note that the additional nodes of the virtual binary trees do not count into the network size but the the levels considered in Line \ref{l:check} of Algorithm \ref{algo:stabembed} correspond to the levels in the virtual tree in order to avoid short-lived re-embeddings.

\longonly{
\textit{Estimated subtree sizes:}
Algorithm \ref{algo:stabembed} relies on the actual sizes of subtrees. In particular, leaves reveal that they have no descendants. Revealing such topology information is potentially undesired in privacy-preserving communication systems such as F2F overlays and might reduce the anonymity. Hence, the subtree size can be obfuscated. For instance, rather than adding 1 for each node, each node can be counted as either 0,1, or 2. We obtain an unbiased estimate as long as the probabilities for 0 and 2 are equal. In order to avoid inferences over time, each node's count should be consistent. 
In our practical evaluation, we consider the impact of using such estimates. 

\textit{Heterogeneous node resources:} If the storage of nodes differs considerable, $\embed$ and $\stab(\embed)$ should consider such heterogeneous resources. We extend the above idea to not necessarily count each node with $1$. Instead, the count of a node and hence the expected assigned content should correspond to a node's resources. In other words, we replace the subtree sizes in Algorithms \ref{algo:stabembed} and \ref{algo:embed} with the overall storage capacity of the nodes in the subtrees. 
A detailed evaluation of the heterogeneous node resources is out-of-scope for this paper due to the lack of realistic models but should be considered in greater detail in the future. 
}

\shortonly{
\subsection{Performance Analysis}

We here state the main result of our theoretical analysis. A full proof can be found in \cite{longvers}. 

\begin{theorem}
Let $(G,\id,C, \add)$ be a content-addressable storage with a tree-based greedy embedding $\id$. The depth of the spanning tree
is at most $\calO(\log n)$. With $S$ denoting the number of siblings of a node, the algorithm $\stab(\embed)$ maintains a
$(\calO(\log n),\delta)$-balanced content-addressable storage at communication complexity
\begin{align*}
\E(cost(\stab(\embed)))=\calO\left(\log^3 n \E(S)\right).
\end{align*}
Using virtual binary trees, the algorithm  $\stab_{virt}(\embed)$ maintains a
$(\calO(\log^2 n),\delta)$-balanced content-addressable storage at communication complexity
\begin{align*}
\E(cost(\stab_{virt}(\embed)))=\calO\left(\log^6 n \right).
\end{align*}
\end{theorem}

}

%% file: analysis.tex
\section{Analysis}
\label{sec:analysis}

We designed a stabilization algorithm $\stab(\embed)$ for a dynamic $(\calO(D), \delta)$-balanced content addressable storage. 
Now, we derive the communication complexity of $\stab(\embed)$.  
\longonly{We start by showing an essential Lemma concerning the expected number of descendants.
We subsequently treat node joins and node departures separately.} 
Throughout the section, we assume that the depth of the tree is at most $D$ and we write $\E(X)$ for the expected value of some described random variable $X$. In particular, we write $\E(S)$ for the expected number of siblings of a random node.
We focus on the ideas of the proofs. The complete proofs can be found in our technical report \cite{longvers}. 

\begin{figure*}[ht]
\subfloat[Stabilization]{\label{fig:oriStab}\includegraphics[width=0.28\linewidth]{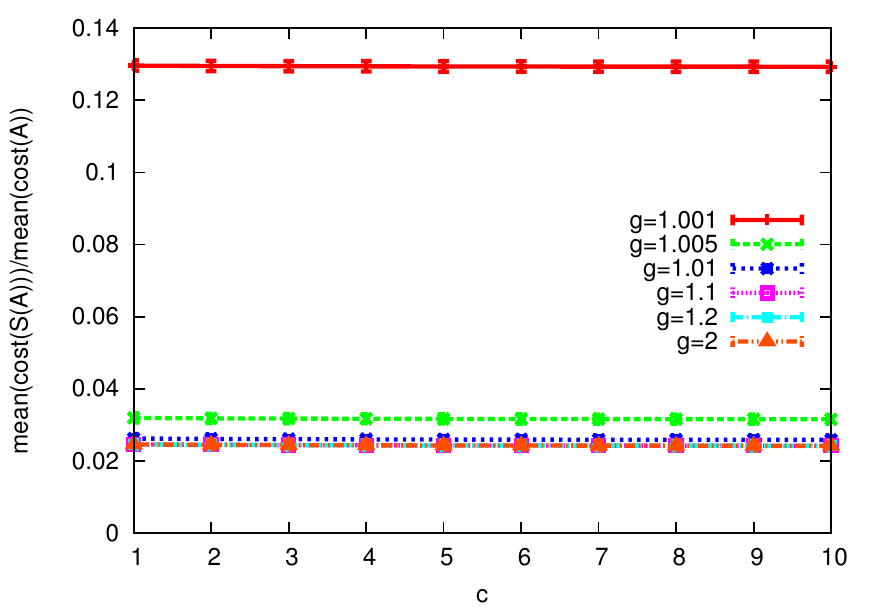}}
\hfill
\subfloat[Mean Imbalance]{\label{fig:oriAvF}\includegraphics[width=0.28\linewidth]{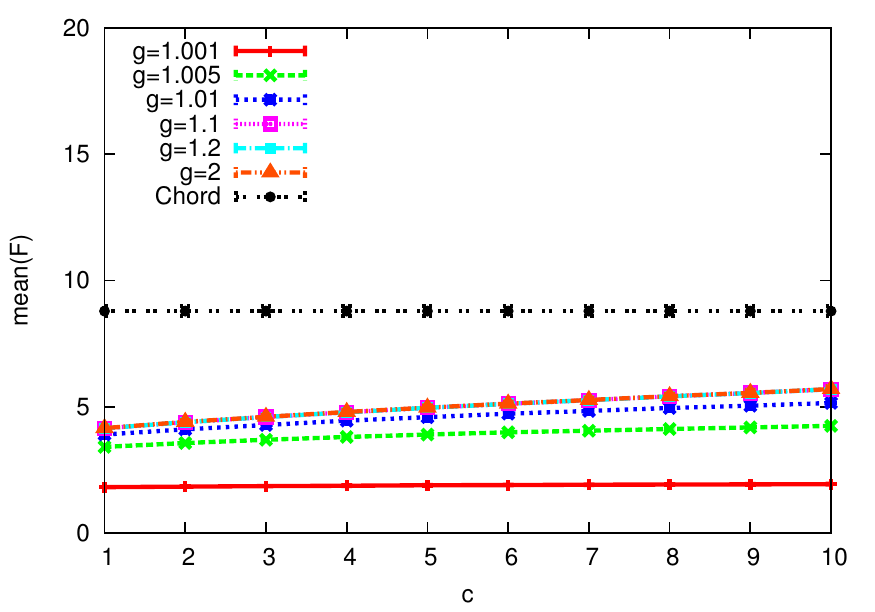}}
\hfill
\subfloat[Max Imbalance]{\label{fig:oriMaxF}\includegraphics[width=0.28\linewidth]{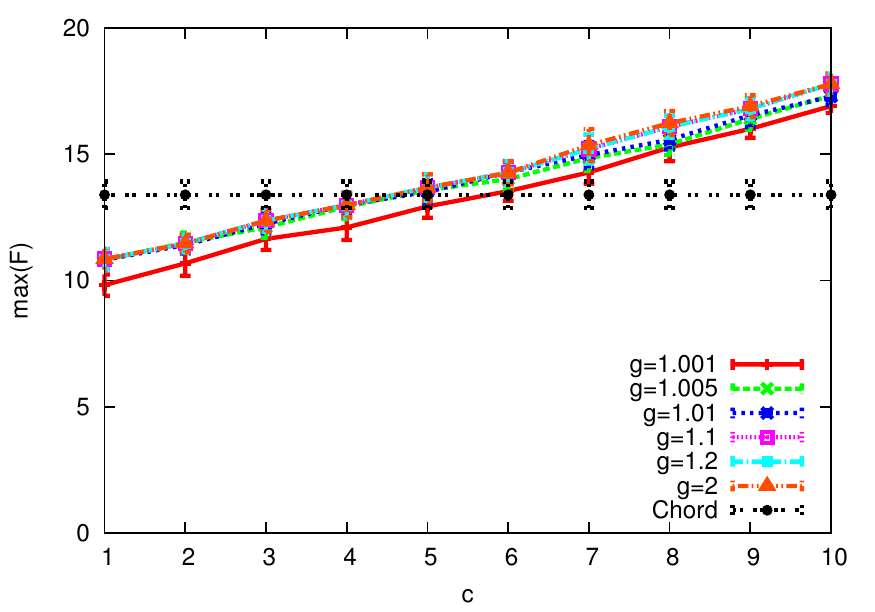}}
\caption{Stabilization overhead (normalized by overhead of complete re-embedding) and corresponding imbalance of the content addressing of Algorithm \ref{algo:stabembed} for various values of $c$ and $g$}
\vspace{-1.5em}
\label{fig:ori}
\end{figure*} 

\longonly{
\begin{lemma}
\label{thm:desc}
The expected number of descendants of a random node is $\E(Y)=\calO(D)$.
\end{lemma}
\begin{proof}
The idea of the proof is to make use of the fact that the average number of descendants is equal to the average number of ancestors. The latter corresponds to the average level of a node and is hence bound by the depth of tree. 
Formally, let $Z=\{(u,v): u \textnormal{ is descendant of }v\}$ denote the set of all descendant-ancestor relations. The expected number of descendants
is hence given by $|Z|/n$.
We now determine an upper bound on $|Z|$.
For this purpose, let $L_u$ denote the level of a node $u$.
A node on level $l$ of the spanning tree is a descendant of $l$ nodes, so that the total number of descendant-ancestor relations
corresponds to the sum of all levels. Hence
\begin{align*}
|Z| = \sum_{u \in V} L_u \leq n |D|.
\end{align*}
Division by $n$ shows the claim. 
\end{proof}
}

\begin{prop}
\label{thm:join}
The communication complexity of Algorithm \ref{algo:stabembed} for a node join is 
$\E\left(cost_{join}(\stab(\embed))\right)=\calO\left(D\right)$
\end{prop}
\shortonly{The idea of the proof is to write the communication complexity as a sum of three parts: i) the complexity of local re-embeddings, ii) the complexity of the network size estimation updates, and iii) the complexity of potential re-embeddings after reaching the root. 
We show that the first is bound by the expected number of descendants of a node, scaling with $\calO(D)$. The second corresponds to forwarding a message to the root along the tree, resulting again in the bound $\calO(D)$. In the case of iii), we show that the probability of starting an embedding from the root is $\calO\left(\frac{1}{n}\right)$, resulting in a constant expected complexity due to the linear overhead of the embedding.}

\longonly{
\begin{proof}
We write the communication complexity as a sum of three phases.
First, $X_1$ denotes the complexity of local re-embeddings, corresponding to Line \ref{l:em1} in Algorithm \ref{algo:stabembed}.
Second, $X_2$ denotes the complexity of propagating status updates to the root (Lines \ref{l:forward}, \ref{l:em2} and \ref{l:relay}).
Third, $X_3$ denotes the communication complexity after the updates have reached the root (Lines \ref{l:rootS}-\ref{l:rootE}).
So, $\E\left(cost_{join}(\stab(\embed))\right)=\E(X_1)+\E(X_2)+\E(X_3)$.

We start by considering $X_1$. 
Note that the parent node $u$ of the newly joined node $v$ calls $\embed$ rather than relaying the request
as the quantity on the left of Eq. \ref{eq:criterion} is reduced and thus remains smaller than $f(u)$. 
Thus, $X_1$ corresponds to the complexity of applying $\embed$ to a subtree consisting of $u$, $v$, and
all of $Y$ descendants of $u$. By the third condition in Definition \ref{def:dymEm} and Lemma \ref{thm:desc},
$\E(X_1)=\calO(2+\E(Y))=\calO(D)$ follows.

The number of propagated updates $X_2$ is bound by the longest path from a node to the root, hence
$\E(X_2)=\calO(D)$.

In order to determine $X_3$, let $E$ denote the event that the new network size $n$ after the join exceeds $n_{est}g$.  
We have $\E(X_3)=\E(X_3|E)P(E)$, because the root only sends additional messages if a new estimate needs to be broadcast.
The complexity of re-embedding and broadcasting $n_{est}$ is $\E(X_3|E)=\calO(n)$.
However,  the event $E$ implies that $n = n_{est}g+1$, hence it only occurs after at least $n_{est}(g-1)+1$ nodes joined.
Hence $E$ occurs only for a fraction $P(E)=\calO\left(\frac{1}{n}\right)$ of the joins.
So, $\E(X_3)=\calO(1)$ and indeed
\longonly{
\begin{align*}
\E\left(cost_{join}(\stab(\embed))\right)=\E(X_1)+\E(X_2)+\E(X_3)=\calO(D).
\end{align*}
}
\shortonly{$\E\left(cost_{join}(\stab(\embed))\right)=\E(X_1)+\E(X_2)+\E(X_3)=\calO(D).$}
\end{proof}
}

\begin{prop}
\label{thm:depart}
The communication complexity of Algorithm \ref{algo:stabembed} for a node departure is 
$\E\left(cost_{depart}(\stab(\embed))\right)=\calO\left(D^3\E(S)\right).$
\end{prop}
\shortonly{As above, the idea of the proof is to divide the overall communication complexity in parts. In addition to the aspects above, we have to consider the rejoins of the departed node's descendants, which we can bound by Proposition \ref{thm:join}. 
The main difficulty lies in determining the complexity of the local re-embeddings. 
We derive the probability that a node has to relay a request for re-embedding in relation to its number of descendants. 
Based on the results, we can determine a global upper bound on the likelihood that a node has to participate in a re-embedding. 
Now, the expected number of nodes that have to participate corresponds to the complexity of the local re-embeds.}

\longonly{
\begin{proof}
Analogously to the proof of Proposition \ref{thm:join}, we derive the desired bound as the sum of four
phases $X_1$, $X_2$, $X_3$, and $X_4$.
The decisive quantity is $\E(X_1)$, which corresponds to local re-embeddings by an ancestor of the departed node as a direct reaction to the departure rather than to a re-join of a descendant of the departed node. 
As in the proof of Lemma \ref{thm:join}, $X_2$ and $X_3$ denote the complexity of propagating updates to and from the root with $\E(X_2)=\calO(D)$ and $\E(X_3)=\calO(1)$.
$X_4$ denotes the complexity resulting from the re-joins of the descendants of the departed node. 
By Lemma \ref{thm:desc} and Lemma \ref{thm:join}, these correspond to an expected number of $\calO(D)$ joins at an expected communication complexity of $\calO(D)$ each, hence $\E(X_4)=\calO(D^2)$. 

The main difficulty lies in deriving $\E(X_1)$. We obtain a global upper bound on the probability $p$ that a node $v$ has to participate in the local re-embedding caused by the departure. Then, we have $\E(X_1)=\calO(n p)$. 
A node $v$ has to participate if one of $v$'s descendants departs or one of $v$'s ancestors re-embeds. By Lemma \ref{thm:desc}, the probability $P(E_1)$ that a descendant departs is $P(E_1)=\calO(D/n)$.
For the second event $E_2$ of being affected by an ancestor's re-embedding, we first consider the frequency of relayed embedding requests.
Let $Z$ denote the number of topology changes until a node $u$ has to relay a request to its parent after the last re-embedding of $u$'s subtree initiated by any ancestor of $u$.
In the following, we show that $\E(Z)=\Omega\left(\frac{n}{D^2}\right)$ regardless of $u$'s position in the tree. Thus, the probability that a node departure results in a re-embedding request from $u$ to its parent is $P(E_3)=\calO\left(\frac{D^2}{n}\right)$ for all nodes $u$.
Now, a node $v$'s ancestor re-embeds if $v$ or one of $v$'s ancestors or their siblings request a re-embed.
$v$ has at most $D$ ancestors with an expected number of siblings of $\E(S)$, hence by a union bound
$P(E_2)=P(E_3)DK=\frac{D^3(\E(S)-1)}{n}$.
Thus, 
\begin{align*}
\E(X_1)=\calO\left(n(P(E_1)+P(E_2))\right)=\calO\left(D^3\E(S)\right).
\end{align*}

It remains to determine $\E(Z)$. 
We start by considering the number of descendants $U$ of $u$ that either depart or have to re-join due to a  departure. 
Then, we derive a lower bound $\theta$ on the number of nodes that can depart before $u$ relays a re-embedding request. It follows $\E(Z)=\theta /\E(U)$.
The network size does not change considerably between a re-embedding from an ancestor and a request for re-embedding, as otherwise the root would initialize a re-embedding. 

The probability that a departing node is one of $u$'s initial descendants is $\calO\left(\frac{|V_u|}{n}\right)$. By Lemma \ref{thm:desc}, a departure affects on average $\calO(D)$ nodes, namely all descendants of the departing node. By conditioning on the fact that the departure takes place in subtree of potentially less than $n$ nodes, $\calO(D)$ remains a valid upper bound. Hence  $\E(U)= \calO\left(\frac{D |V_u|}{n}\right)$. 

After a re-embedding initiated by $p(u)$ or another ancestor, we know that the content addressable storage is $(f(p(u)),\delta)$-balanced with regard to $V_u$.
Hence, by Eq. \ref{eq:fu} and $level(p(u))=level(u)-1$, $con(V_u)=\sum_{v \in V_u}\mu(\closest(v))\leq |V_u|(\frac{g(1+c+level(u)-1)}{n_{est}g} + \delta)$ or 
\begin{align}
\label{eq:relVu}
|V_u|\geq \frac{cont(V_u)n_{est}g}{(1-\frac{g}{f(u)})f(u)+n_{est}g\delta}
\end{align}
By Eq. \ref{eq:criterion}, $u$ has to request a re-embedding if the subtree size $|V'_u|$ has become so small that
\begin{align}
\label{eq:relVu2}
|V'_u|\leq \frac{cont(V_u)n_{est}g}{f(u)}
\end{align}
As $\delta$ is negligible, $f(u)=\calO(D)$, Eq. \ref{eq:relVu} and Eq. \ref{eq:relVu2} result in $|V'_u|=\calO\left(1-\frac{1}{D}|V_u|\right)$.
In other word, at least $\theta=\Omega\left(\frac{1}{D}|V_u|\right)$ descendants of $u$ have to depart or re-join such that $u$ has to re-embed. Hence,
\begin{align*}
\E(Z)\geq\frac{\theta}{\E(A)}=\Omega\left(\frac{|V_u|}{D}\frac{n}{D|V_u|} \right)= \Omega\left(\frac{n}{D^2} \right)
\end{align*} 
and thus indeed $\E(X_1)=\calO(D^3\E(S))$ and $\E\left(cost_{depart}(\stab(\embed))\right)=\calO\left(D^3\E(S)\right)$. 
\end{proof}
}
\longonly{
We now obtain the general result for a tree of logarithmic depth, considering both the original algorithm $\stab(\embed)$ and the version $\stab(\embed)_{virt}$ relying on virtual binary trees. For the latter, the depth is bound by $\calO(D \log n)$ rather than $D$ but the expected number of siblings is $\E(S)\leq 1$. 
\begin{corollary}
Let $(G,\id,C, \add)$ be a content-addressable storage with a tree-based greedy embedding $\id$. The depth of the spanning tree
is at most $\calO(\log n)$. With $S$ denoting the number of siblings of a node, the algorithm $\stab(\embed)$ maintains a
$(\calO(\log n),\delta)$-balanced content-addressable storage at communication complexity
\begin{align*}
\E(cost(\stab(\embed)))=\calO\left(\log^3 n \E(S)\right).
\end{align*}
Using virtual binary trees, the algorithm  $\stab_{virt}(\embed)$ maintains a
$(\calO(\log^2 n),\delta)$-balanced content-addressable storage at communication complexity
\begin{align*}
\E(cost(\stab_{virt}(\embed)))=\calO\left(\log^6 n \right).
\end{align*}
\end{corollary}
Thus, if the trees are reasonable regular, i.e., $\E(S)$ is bound by a constant or a (poly-)logarithmic factor, the original version $\stab(\embed)$ exhibits both a better balance and lower computation complexity. However, if $\E(S)$ is high, i.e., there are few nodes with many children while the majority of nodes have few children, the computation complexity increases. Then, the virtual binary tree version can offer a lower computation complexity at the price of a higher balance factor.    

In the following, we simulate the actual overhead of the two algorithms and compare the results to these bounds. 
}
\shortonly{Hence, for a logarithmic tree depth, we have a stabilization complexity of $\calO(\log^3 n ~\E(S))$, meaning that for a polylog number of siblings we indeed have $\calO(polylog(n))$. If we use virtual binary trees, we have $\E(S)=2$ in the virtual tree. However, the tree depth $D$ increases by up to a factor $\calO(\log n)$ (the maximal depth of a virtual binary tree at any of the original $D$ levels), resulting in a stabilization complexity of $\calO(\log^6 n)$. The corresponding bounds on $f$ are $\calO(\log n)$ for the original version and $\calO(\log^2 n)$ for virtual trees.}

%% file: eval.tex
\section{Simulations}
\label{sec:eval}

We substantiate the previous asymptotic bounds by a simulation study considering the case of F2F overlays.
Our goal is to provide concrete bounds on the stabilization complexity and the balance of the content addressing of Algorithm \ref{algo:stabembed} for various values of the tree depth offset $c$ and network size estimation accuracy $g$.  
Furthermore, we evaluate the impact of the simple join and virtual binary tree variant described in Section \ref{sec:vary}. 
In the following, we describe our simulation model, set-up, expectations and results.

\emph{Simulation Model and Set-up:} Our simulation builds on GTNA \cite{schiller2013gtna}, a framework for graph analysis. 
Aside from the parameters $c$ and $g$, the performance of Algorithm \ref{algo:stabembed} depends on the graph $G$ and the churn pattern, i.e., the node join and departure sequence.
In our simulation model, we characterize the latter by the session and intersession length distributions $L_{S}$ and $L_{I}$. 
Furthermore, we use the spanning tree construction by Perlman \cite{Perlman85}, which assigns each node a random numerical identifier and then constructs a spanning tree of minimal depth such that the root corresponds to the node with the highest identifier.
 
During the set-up phase, each node chooses its random identifier for the spanning tree construction, which remains constant during the simulation. 
Initially, each node is online with probability $\frac{\E(L_{S})}{\E(L_{S}+L_{I})}$. If a node is online at start-up, we assume that it is at a random point of its current session. In other words, we choose the time of an online node's departure by selecting a session length $l$ according to $L_S$ and multiplying $l$ with a uniformly chosen random number in $[0,1)$. Analogously, we select the time until an offline node joins. 
Then, we execute the spanning tree construction on the subgraph induced by all initially online nodes. Subsequently, we execute Algorithm \ref{algo:embed} on the same subgraph to obtain the initial embedding.
If the graph is partitioned into multiple components, we execute the two algorithms for each component individually. 

In each step $i$ of the algorithm, we add or remove a node according to the previously selected sessions and intersession times. 
We choose the time of this node's next join or departure by selecting an interval $l$ according to $L_{I}$ or $L_{S}$, respectively, and add $l$ to the currently elapsed time. 
Afterwards, we re-establish the spanning tree, potentially merging trees if previously disjoined components are connected or constructing new trees if new partitions are created. 
Last, we execute Algorithm \ref{algo:stabembed}, starting from either a newly joined node or the parent and children of a departed node.

\begin{figure*}[ht]
\centering
\subfloat[Stabilization]{\label{fig:varStab}\includegraphics[width=0.35\linewidth]{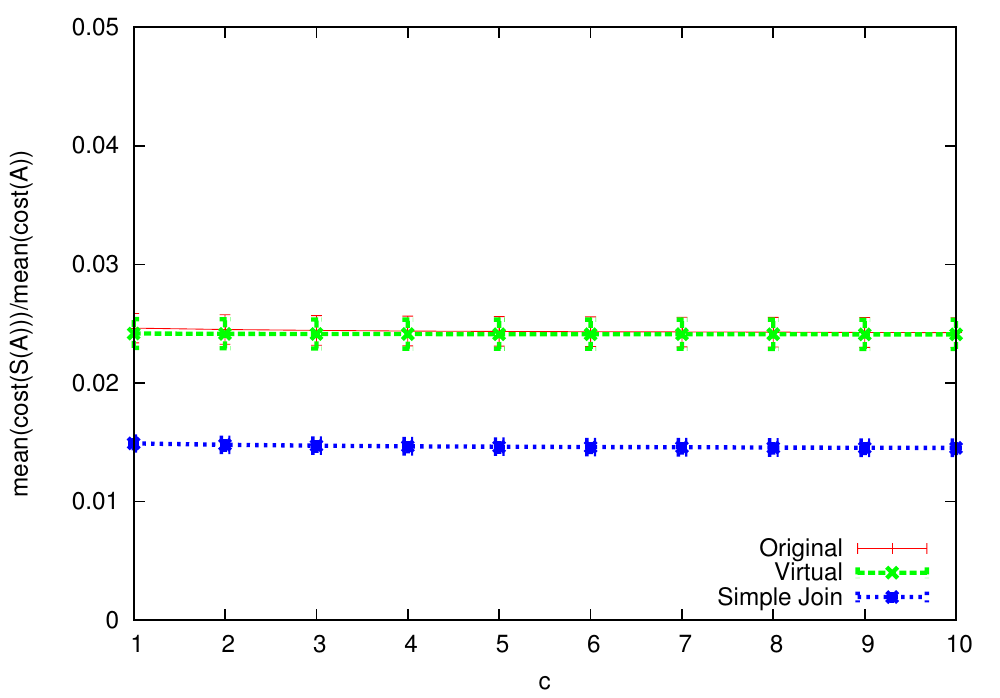}}
\subfloat[Imbalance]{\label{fig:varF}\includegraphics[width=0.35\linewidth]{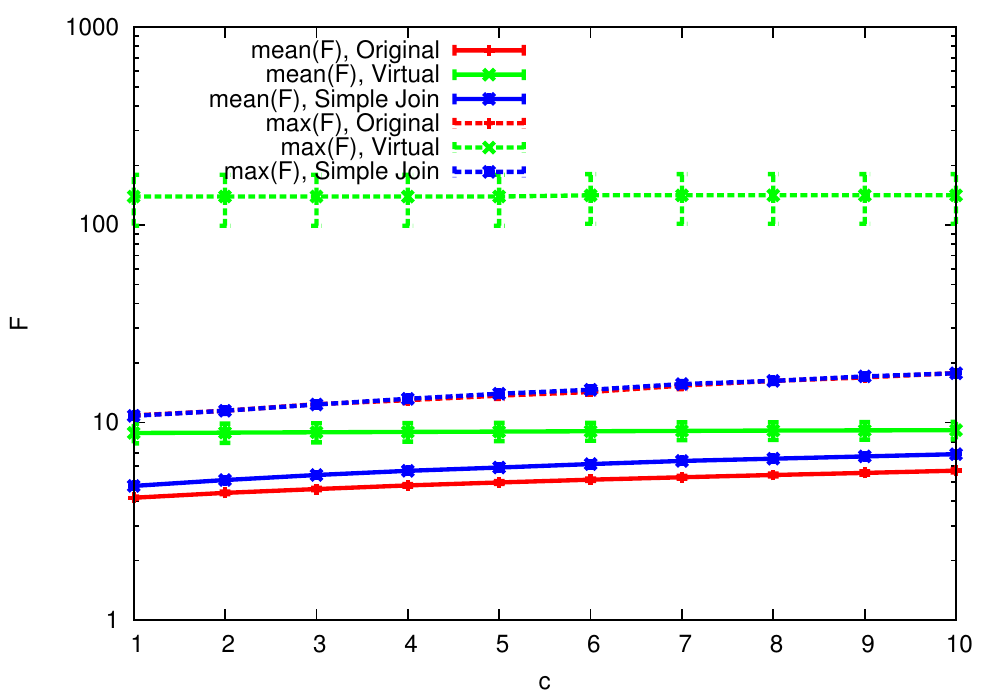}}
\caption{Stabilization overhead and corresponding imbalance of the content addressing of Algorithm \ref{algo:stabembed} in this original form, the virtual binary tree variant, and the simple join variant, $g=2$}
\label{fig:variants}
\vspace{-2em}
\end{figure*}

During the simulation, we measure the number of all messages $cost_i(\stab(\embed))$ required for stabilization at step $i$.
Let $mean(cost(\stab(\embed)))$ denote the mean of  $cost_i(\stab(\embed))$ over all $i$.
Furthermore, we compare the cost of our algorithm to that of re-embedding at each topology change, i.e., we compute $\frac{mean(cost(\stab(\embed)))}{mean(\embed)}$ with $mean(\embed)$ being the mean number of messages to i) inform the root of the joined or departed node's tree of the change, and ii) executing Algorithm \ref{algo:embed} on the complete tree.
\longonly{The quantity $\frac{mean(cost(\stab(\embed)))}{mean(\embed)}$ thus indicates how much our changes reduce the stabilization complexity.}
In order to characterize the balance of the content addressing, we consider the fraction of addresses $\mu_i(u)$ assigned to each online node $u$ in step $i$ in relation to the number of nodes $n_i(u)$ in $u$'s component, i.e., we derive $F_{u,i}=\mu_i(u)\cdot n_i(u)$. 
For each step $i$, we derive the maximum $F_i = \max_{u: online}F_{u,i}$ and compare it to the upper bound on the maximal permitted imbalance defined in Eq. \ref{eq:fu}. We then consider the mean and maximal imbalance $mean(F)$  and $max(F)$ over all steps $i$.

Our sample set-up considers the case of F2F overlays, which are route-restricted and limit direct communication to devices of users with a mutual trust relationship. 
In this study, we use the friendship graph of a university online social network (SPI) of $9,222$ students with an average of $10.58$ connections to model trust relations \cite{paul2015students}. 
\longonly{SPI represents a friendship network and is thus in the absence of actual F2F overlay topologies a suitable model for such an overlay. In contrast to subgraphs of Facebook or other large-scale social networks, the locality of the network indicates that people sharing a link indeed share a real-world friendship or at least acquaintance.}
In order to judge the impact of the graph topology on the algorithms, we also generated one synthetic graph according to the model of Barabasi-Albert (BA) and another graph according to the model of Erdos-Renyi (ER-1), both with the same number of nodes and the same average degree of $10.58$ as SPI.
Furthermore, we generated another Erdos-Renyi graph with 9,222 nodes but a higher average degree of $922.2$ (ER-2), to shed light on the impact of the density of the network. 
Our churn patterns follow the empirical session and intersession length measured in Freenet, an anonymous content sharing network with a F2F mode \cite{roos2014measuring}. Based on these churn patterns, the number of concurrently online nodes usually varies between $3,700$ and $4,000$. 
As for the parameters of Algorithm \ref{algo:stabembed}, we varied $c$ between $1$ and $10$ and choose $g \in \{1.001, 1.005, 1.01, 1.1, 1.2, 2\}$.
All parameter combinations were considered for the original form of Algorithm \ref{algo:stabembed} as well as for the simple join and the virtual binary tree variant.
We averaged our result over 20 runs and present them with 95\% confidence intervals. Each run consists of 100,000 consecutive node joins or departures.
Note that we used the same 20 joins and departure sequences for each set of parameters in order to facilitate comparisons.
For comparison, we also measure $mean(F_i)$ and $max(F_i)$ for a Chord overlay \cite{stoica2001chord} of the same size using these join and departure sequences.

\emph{Expectations:} Our expectations with regard to the parameters $c$ and $g$ on the stabilization overhead and the balance of the content addressing are governed by Eq. \ref{eq:fu} and Propositions \ref{thm:join} and \ref{thm:depart}.
Eq. \ref{eq:fu} indicates that the upper bound on the maximal imbalance increases with both $g$ and $c$.
However, considering Line \ref{l:check} of Algorithm \ref{algo:stabembed}, we see that $g$ does not affect the actual decision of re-embedding. Rather, it only affects the certainty of nodes in the current network size estimation and thus has at most an indirect effect on the actual balance $F_i$.
In contrast, $c$ affects the decision in Line \ref{l:check} and allows nodes to accept a larger imbalance. Hence, we expect an increase in $mean(F)$ and $max(F)$ with an increased $c$.

With regard to the stabilization complexity, the asymptotic bounds in Propositions \ref{thm:join} and \ref{thm:depart} are independent of both $g$ and $c$. Indeed, $g$ only indicates the frequency of re-embeddings due to an inaccurate network size estimation. For larger $g$, such re-embeddings should be rare, so that we do not expect a considerable impact of $g$ on either the actual imbalance or the stabilization overhead. 
However, for very low $g=1.001$ or $g=1.005$, re-embeddings only require a change of less than 1\% in size, so that re-embeddings actually impact the overall stabilization overhead. Thus, we assume that the stabilization overhead is higher for these $g$ whereas $mean(F)$ and $max(F)$ decrease due to the frequent re-embeddings, which re-establish a perfectly balanced address assignment.

The goal of using a simplified join mechanism and virtual binary trees is to reduce the stabilization overhead. While Proposition \ref{thm:depart} indicates a reduced stabilization overhead when using virtual binary trees, the increased depth of the binary trees entails an increased upper bound on the permitted imbalance as by Eq. \ref{eq:fu}.
Thus, it is likely that the actual values $mean(F)$ and $max(F)$ are higher than for Algorithm \ref{algo:stabembed}.
Similarly, our simplified join mechanism reduces the stabilization overhead by not requiring re-embeddings of the complete subtree rooted at the parent. \longonly{However, such re-embeddings after joins present the opportunity to improve the balance  if previous departures have already changed the content assignment within the subtree considerably but not sufficiently for a re-embedding.
Thus, the price for the reduced stabilization complexity is likely to be an increased actual imbalance though the bound on the permitted imbalance remains unaffected.}

We choose the four types of graphs (SPI, BA, ER-1, ER-2) in order to ascertain that our expectations with regard to the impact of the node degree hold. 
For instance, we expect that the low-degree random graph ER-1 results in trees with a low number of children per node and hence a high depth.
By Eq. \ref{eq:fu}, the high depth should correlate with a high imbalance. 
Due to the low number of children, we expect a comparable low stabilization overhead for ER-1.
For analogous reasons, we expect the opposite results, namely a high stabilization overhead and a well-balanced content addressing, for ER-2.
The results for BA and SPI should moderate between those of the two random graphs.

\emph{Results:}
Our results with regard to the impact of parameters $c$ and $g$ agree with the above expectations and underline the asymptotic bounds with concrete values.
Fig. \ref{fig:ori} 
displays the results for the original version of Algorithm \ref{algo:stabembed} on the SPI graph. Notably, our algorithm reduces the stabilization overhead to 2-3\% of a complete re-embedding for $g\geq 1.01$, as shown in Fig. \ref{fig:oriStab}. In absolute numbers, the average number of messages sent per step is slightly above 80 for a network of 3,000 to 4,000 online nodes. As expected, very low values of $g$ considerably increase the stabilization overhead because the network size estimation has to be adjusted frequently. 
Fig. \ref{fig:oriAvF} shows that there are nodes in the network that are responsible for 4 to 6 times as many addresses as the average node in their component for $g\geq 1.01$. For $g=1.001$, $mean(F)$ can be as low as $1.82$ at the price of a high stabilization overhead. 
In the worst case, displayed in Fig. \ref{fig:oriMaxF}, the imbalance increase to up to a factor 20. Note that the depth of the spanning tree varies between 20 and 30, so that the observed maximum is usually considerably lower than the theoretical upper bound.
While the stabilization overhead is indeed not significantly impacted by $c$ and $g$, an increase of $c$ entails an increase in imbalance, as expected. 
In comparison to Chord, a supposedly well-balanced P2P overlay, we achieve a lower value of $mean(F)$ for all considered parameters.
With regard to $max(F)$, we achieve a higher degree of balance for $c < 5$.
Hence, our content addressing achieves a similar or even better balance than existing solutions for content addressing. 

\begin{table}
\centering 
\begin{small}
\begin{tabular}{||l||c|c||c|c||c|c||}
\hline 
Sys & \multicolumn{2}{c||}{Original} & \multicolumn{2}{c||}{Simple} & \multicolumn{2}{c||}{Virtual} \\
 \hline 
 SPI & 0.025 & 4.16 & 0.015 & 4.78 & 0.024 & 8.84 \\ 
 BA & 0.041 & 4.20 & 0.021 & 4.79 & 0.040 & 11.33 \\
 ER-1 & 0.023 & 7.72 & 0.020 & 8.58 & 0.022 & 27.91 \\
 ER-2 & 0.071 & 1.14 & 0.036 & 1.21 & 0.071 & 1.14 \\ \hline
\end{tabular} 
\end{small}
\caption{Stabilization overhead $\frac{mean(\stab(\embed))}{mean(\embed)}$ (left column) vs. mean imbalance factor $mean(F)$ (right column) of Alg. \ref{algo:stabembed} ($c=1$, $g=2$) in different topologies: real-world social network SPI; Barabasi-Albert (BA); Erdos-Renyi, average degree $10.58$ (ER-1); Erdos-Renyi, average degree $922.2$ (ER-2)}
\label{tab:topologies}
\vspace{-2em}
\end{table} 

Now, we consider the impact of protocol variants on stabilization and content addressing. 
Fig. \ref{fig:variants} contrasts virtual binary trees and simple join with the original algorithm for $g=2$ and $c=1..10$.
Indeed, both variants decrease the stabilization overhead further, as displayed in Fig. \ref{fig:varStab}.
However, the insignificantly decreased stabilization overhead for the virtual tree variant comes at the price of a considerably higher imbalance.
As the depth of the virtual tree is usually between 200 and 300, Fig. \ref{fig:varF} shows that the actual observed imbalance can reach values close to 200.
In contrast, a simple join only slightly increases $mean(F)$ but leaves $max(F)$ largely unaffected.
Thus, the simulation study indicates that the actual improvement of the virtual binary tree variant with regard to stabilization overhead does not outweigh the drastic decrease in balance, at least for the considered graphs. 
However, using a simple join mechanism reduces the stabilization overhead by roughly a factor 2 without severe consequences on the balance of the content addressing.     

Last, Table \ref{tab:topologies} displays $\frac{mean(\stab(\embed))}{mean(\embed)}$ and $mean(F)$ for the four considered topologies focusing on the case of $c=1, g=2$.  The structure of the underlying graph drastically impacts the actual results. The stabilization overhead increases drastically for the densely connected network ER-2. Note that the virtual binary trees cannot counteract this increase, as the re-embedding is almost always executed by the parent of the newly joined or departed node. Thus, a simple join variant indeed nearly halves the overhead. Due to the low depth of the tree, ER-2 exhibits an extremely low imbalance $mean(F)$. 
In contrast, ER-1 exhibits a very low stabilization overhead but a considerably higher imbalance, as expected due to the low degree and high tree depth. 
BA and SPI, having a scale-free degree distribution with some high-degree nodes and mostly low-degree nodes, moderate between the extremes.    
 
This evaluation complements our theoretical bounds with concrete numbers. Not only do these concrete numbers validate our theoretical bounds, they also indicate that our algorithm is efficient and can achieve a more balanced content addressing than commonly used content addressing schemes.   

%% file: conclusion.tex
\section{Conclusion}
\label{sec:conclusion}

\longonly{
The main contribution of this paper is the design and formal verification of an approach that efficiently generates tree-based greedy embeddings for balanced content addressing on fully dynamic networks. We realized our solution by designing a stabilization algorithm and an embedding algorithm where the former makes use of the latter to dynamically update content network addresses.


We proved that our approach guarantees fair distribution of content addresses and we showed that the expected cost of a single change in the network is logarithmic for node joins and polylogarithmic for node departures. Finally we confirmed these formal bounds in a simulation on realistic problem instances.

Future work may explore alternative implementations of our algorithms to potentially improve upon the complexity bounds and the real world latency of both routing and stabilization.
}

\shortonly{
The main contribution of this paper is the design and formal verification of an approach that efficiently generates tree-based greedy embeddings for balanced content addressing on fully dynamic networks. A complementary simulation-based evaluation indicates that our stabilization algorithm is efficient while our content addressing scheme achieves a similar or improved degree of balance as DHTs. In contrast to DHTs, our algorithms are suitable for route-restricted networks without the freedom to establish (overlay) connections between arbitrary nodes.
Future work may explore the question of heterogeneous node resources as well as address the issue of resilience to failures and malicious behavior. 

}
